\documentclass{llncs}

\usepackage{amsmath, amssymb}
\usepackage{algorithm}
\usepackage[noend]{algpseudocode}
\usepackage{graphicx}
\usepackage{fullpage,authblk,prettyref}

\newcommand{\bs}{\setminus}
\newcommand{\ignore}[1]{}

\newcommand{\NP}{\ensuremath{\mathsf{NP}}}
\newcommand{\PP}{\ensuremath{\mathsf{P}}}

\newcommand{\APX}{\ensuremath{\mathsf{APX}}}
\newcommand{\Z}{\mathbb{Z}}

\newcommand{\R}{\mathbb{R}}

\newcommand{\E}{E}
\newcommand{\F}{\mathcal{F}}

\newrefformat{theoremm}{Theorem \ref{#1}}
\newrefformat{section}{Section \ref{#1}}
\newrefformat{eq}{Equation \eqref{#1}}
\def\xtheorem[#1][#2][#3]{\newtheorem{#2}[theoremm]{#3} \newrefformat{#2}{#3 \ref{#11}}}
\def\ytheorem[#1][#2][#3]{\newtheorem{#2}{#3} \newrefformat{#2}{#3 \ref{#11}}}
\xtheorem[#][observation][Observation]
\xtheorem[#][eg][Example]
\xtheorem[#][fact][Fact]
\xtheorem[#][app][Appendix]
\xtheorem[#][alg][Algoritheorem]
\xtheorem[#][step][Step]

\newcommand{\comment}[1]{\hfill \texttt{//} #1}

\newcommand{\qedhere}[0]{\qed}

\begin{document}
\mainmatter

\title{Generalized Hypergraph Matching via Iterated Packing \\and Local Ratio}

\author{Ojas Parekh\inst{1} and David Pritchard\inst{2}}

\institute{Sandia National Laboratories, Albuquerque NM 87185, USA\footnote{Sandia National Laboratories is a multi-program laboratory managed and operated by Sandia Corporation, a wholly owned subsidiary of Lockheed Martin Corporation, for the U.S. Department of Energy's National Nuclear Security Administration under contract DE-AC04-94AL85000.},\\
\email{odparek@sandia.gov}
\and
University of Southern California, Los Angeles CA 90089, USA,\\
\email{dpritcha@usc.edu}}




\date{}
\maketitle

\begin{abstract}
In $k$-hypergraph matching, we are given a collection of sets of size at most $k$, each with an associated weight, and we seek a maximum-weight subcollection whose sets are pairwise disjoint. More generally, in $k$-hypergraph $b$-matching, instead of disjointness we require that every element appears in at most $b$ sets of the subcollection. Our main result is a linear-programming based $(k-1+\tfrac{1}{k})$-approximation algorithm for $k$-hypergraph $b$-matching. This settles the integrality gap when $k$ is one more than a prime power, since it matches a previously-known lower bound. When the hypergraph is bipartite, we are able to improve the approximation ratio to $k-1$, which is also best possible relative to the natural LP. These results are obtained using a more careful application of the \emph{iterated packing} method.

Using the bipartite algorithmic integrality gap upper bound, we show that for the family of combinatorial auctions in which anyone can win at most $t$ items, there is a truthful-in-expectation polynomial-time auction that $t$-approximately maximizes social welfare.  We also show that our results directly imply new approximations for a generalization of the recently introduced bounded-color matching problem.

We also consider the generalization of $b$-matching to \emph{demand matching}, where edges have nonuniform demand values. The best known approximation algorithm for this problem has ratio $2k$ on $k$-hypergraphs. We give a new algorithm, based on local ratio, that obtains the same approximation ratio in a much simpler way.

\end{abstract}

\section{Introduction}\label{sec:theintro}
In a matching problem we want to find the maximum weight subcollection of pairwise disjoint sets within a given collection. Often these problems are studied with respect to the maximum set size $k$ (i.e.~on ``$k$-hypergraphs"); matching is polynomial-time solvable for $k=2$, while it is \APX-hard for $k=3$, even in special cases like \emph{3-dimensional matching}~\cite{Kann91}.

The $b$-matching problem generalizes matching: the input specifies a limit $b_v$ for every vertex, and we can select at most $b_v$ sets containing each $v$; ordinary matching results when $b$ is the all-1 vector. A $b$-matching instance can allow each set $e$ to be selected multiple times up to some upper \emph{capacity} limit $c_e$. \emph{Simple} $b$-matching is the case where all capacities are unit. The \emph{uncapacitated} case is where $c = \overrightarrow{\infty}$, i.e.~there are no capacity limits.

One of our results considers the generalization of $b$-matching to \emph{demand matching}, a notion originally introduced for graphs in~\cite{ShepVet-MOR07}. For this problem each edge is given a demand value $d_e$, and we now constrain that for every vertex $v$, the sum of the $d$-values of the incident edges should be at most $b_v$. When $d$ is the all-1 vector we recover the $b$-matching problem.

Hypergraphic matching problems are often studied via linear programming relaxations. In this paper we use only the naive LP relaxations. 
The worst-case ratio between the LP optimum and the optimal integral solution is called the \emph{integrality gap}. 
An \emph{LP-relative $\alpha$-approximation algorithm} is one that produces (in polynomial time) an integral solution of value at least $1/\alpha$ times the LP's optimal value --- this both upper bounds the integrality gap by $\alpha$ and gives an $\alpha$-approximation algorithm. Many classical approximation algorithms are LP-relative; so the notion is not novel, rather, this terminology helps us be concise.

\subsection{Results}

Our main result is the following theorem.
\begin{theorem}\label{theorem:1}
There is an LP-relative $(k-1+\tfrac{1}{k})$-approximation algorithm for $k$-hypergraph $b$-matching, for any capacities.
\end{theorem}
\noindent In \cite{Parekh11} one of the authors announced, without a proof, a weaker result than the above theorem, namely an upper bound of $k-1+\tfrac{1}{k}$ on the integrality gap.  Here we give an algorithm to find an integral solution matching this bound in polynomial time, requiring a significant extension of the techniques presented in \cite{Parekh11}.

For the special case $b=1$, F{\"u}redi, Kahn and Seymour~\cite{KahnEtAl-Combinatorica93} proved an upper bound of $k-1+\frac{1}{k}$ in 1993, while Chan \& Lau~\cite{ChaLau11} recently gave the first polynomial-time algorithm matching this bound.  Their technique does not directly extend to the $k$-hypergraph $b$-matching case.  The technique that we use to prove Theorem~\ref{theorem:1} is \emph{iterated packing}, the same technique from~\cite{Parekh11}. Part of the contribution of the present paper is to simplify and extend some of the the approaches from \cite{Parekh11} and \cite{ChaLau11}. Our main technical innovation is, using iterated packing, to explicitly specify particular additional solutions as ineligible for packing: not only solutions that would be ineligible for the \emph{original} problem, rather we additionally prohibit solutions exceeding the ceiling of the current fractional solution. 


Theorem~\ref{theorem:1} is tight for infinitely many $k$: when $k-1$ is a prime power, as observed in \cite{KahnEtAl-Combinatorica93}, the projective plane $\mathsf{PG}(2, k-1)$ of order $k-1$ yields a matching lower bound of $k-1+\tfrac{1}{k}$ on the integrality gap. It is an interesting open question to settle the integrality gap for any other values of $k$.

We are able to determine the exact integrality gap for another interesting class of hypergraphs. Call a hypergraph \emph{bipartite} (\cite{Aharoni90}; cf.~\cite{PS09}) if, for some distinguished subset $U$ of vertices, every hyperedge contains exactly one vertex from $U$.  
\begin{theorem}\label{theorem:2}
There is an LP-relative $(k-1)$-approximation algorithm for \emph{bipartite} $k$-hypergraph $b$-matching, for any capacities.
\end{theorem}
\noindent Chan and Lau~\cite{ChaLau11} proved Theorem~\ref{theorem:2} in the special case that $b=1$ and the instance is $k$-dimensional\footnote{A hypergraph is $k$-dimensional if for some $k$-partition of the ground set, every edge intersects every part exactly once.}. Proving Theorem~\ref{theorem:2} is similar to Theorem~\ref{theorem:1} plus extending an observation of \cite{ChaLau11} from $k$-dimensional hypergraphs to bipartite ones. Like Theorem~\ref{theorem:1},
a matching integrality gap lower bound is known~\cite[p.~157]{Furedi81} when $k-1$ is a prime power: the hypergraphic dual of the affine geometry $\mathsf{AG}(2, k-1)$, i.e. a truncated projective plane,
has integrality gap $k-1$.

We obtain the following interesting corollaries from the bipartite case.  In the bounded-color $k$-hypergraph $b$-matching problem we are given an instance of the $k$-hypergraph $b$-matching problem along with a partition of the edge set into $l$ color classes, $E = E_1 \cup \cdots \cup E_l$, and a positive integer $w_i$ for $1 \leq i \leq l$.  We seek a feasible $k$-hypergraph $b$-matching of maximum weight such that at most $w_i$ edges from class $E_i$ are selected for each $i$.
\begin{corollary}\label{cor:BCHM}There is an LP-relative k-approximation for bounded-color $k$-hypergraph $b$-matching.
\end{corollary}

\begin{corollary}\label{cor:auctions}
For combinatorial auctions where each bidder can win at most $(k-1)$ items, there is a randomized polynomial-time mechanism that, in expectation, is both truthful and $(k-1)$-approximately maximizes social welfare.
\end{corollary}
\noindent
We are not aware of any prior results for this extremely natural class of combinatorial auctions, cf.~\cite[Ch.~12]{NREV07}.

The proof of Corollary~\ref{cor:auctions} uses the mechanism of Lavi and Swamy~\cite{LS05}, where the distinguished vertices in the bipartite hypergraph correspond to the bidders. For this application, it is crucial that Theorem~\ref{theorem:2} gives an \emph{LP-relative} approximation in polynomial time. 

Finally, we give a new short proof of the following known theorem:
\begin{theorem}[\cite{Parekh11}]\label{theorem:3}
There is an LP-relative $2k$-approximation for $k$-hypergraph demand matching.
\end{theorem}
\noindent Our simpler proof is based on the local ratio method, rather than the iterated packing used in \cite{Parekh11}. We rely on a connection in \cite[p.~12]{BarYehudaEtAl-SICOMP06} between local ratio and iterated packing. 

\subsection{Related Work}\label{sec:intro}
As Tutte observed~\cite{Tutte54}, both in edge-weighted graphs and in the cardinality case, uncapacitated graphic $b$-matching can be reduced to matching by replacing each vertex by $b_v$ clones. Each edge $uv$ is likewise cloned $b_ub_v$ times.
This reduction has two problems: (1) the clones cause an exponential increase in the instance size (from $\lg \lVert b \rVert_1$ to $\lVert b \rVert_1$); and (2) it does not work in the capacitated case, since we need to prevent too many clones of the same edge from being selected. Cloning applies to hypergraphs, too, but has the same two problems. Algorithmically, we can often avoid (1) by not dealing with the clones explicity. For graphs we can fix problem (2): an edge-trisecting reduction~\cite{Tutte54}~(see also \cite[p.~562]{Sc03}) extends cloning to work on capacitated instances. But for hypergraphs, there is no known workaround for problem (2).

As a strawman, let us mention that one can reduce capacitated $b$-matching in $k$-hypergraphs to uncapacitated $b$-matching in $(k+1)$-hypergraphs, by inserting new vertices in each hyperedge and by moving each edge's capacity to the $b$ value of its new vertex. One can even then apply cloning. But this is not that useful for us: e.g., we cannot use the previously-known $b=1$ case of version Theorem~\ref{theorem:1} to even prove the nonconstructive version of Theorem~\ref{theorem:1} for general $b$, since this reduction increases the hyperedge size from $k$ to $k+1$.


Algorithmically, the \emph{simple} (capacity $c=\mathbf{1})$ case of $b$-matching is the hardest. The proof is standard, by fixing the integer part of an optimal fractional solution.
\begin{observation}\label{obs:caps}
Given an (LP-relative) $\alpha$-approximation to simple $b$-matching in $k$-hypergraphs, we can obtain the same quality of approximation for general capacities.
\end{observation}

\paragraph{Hypergraph matching.} Matching problems in $k$-uniform hypergraphs are well-studied algorithmically. For any fixed $\varepsilon > 0$ the best known approximation ratios are $\frac{k}{2} + \varepsilon$ for the unweighted version by Hurkens and Schrijver~\cite{HurSch89} and $\frac{k+1}{2} + \varepsilon$ for the weighted version by Berman~\cite{Ber00}.  In the case $k=3$, the algorithmic results of \cite{ChaLau11} give an $\varepsilon$-improved approximation ratio of 2 for 3-dimensional matching.
On the other hand, Hazan, Safra and
Schwartz~\cite{HazanEtAl-CC06} showed that the problem is hard to approximate within a factor of $\Omega(\frac{k}{\log k})$ unless $\PP = \NP$, even in the $k$-dimensional case.

\paragraph{Hypergraph $b$-matching.} For $b$-matching in $k$-hypergraphs, Krysta~\cite{Krysta-MFCS05} gave a greedy $k+1$-approximation for the simple case, and Young \& Koufogiannakis~\cite{YoungKouf-DISC09} gave a $k$-approximation for the uncapacitated version. Both of these approximation algorithms give LP-relative guarantees. An improvement in some cases was recently obtained by the \emph{$k$-exchange system} framework of Feldman et al.~\cite{FNSW11}. The $b$-matchings form a \emph{$k$-exchange system} (this is explicit only for $k=2$ in \cite{FNSW11}). In this way one can obtain a local search-based $(\frac{k+1}{2}+\varepsilon)$-approximation algorithm for weighted $k$-hypergraph $b$-matching. However, its running time is exponential in $k$ and it does not give any LP-relative guarantee.

It may be tempting to think that the $b$-matching problem in hypergraphs is a simple extension of $1$-matching in hypergraphs because the theory and algorithms for $b$-matching in graphs closely relate to those for $1$-matchings.  As evidenced by the results above, this does not appear to be the case.  An approximation algorithm that runs in time polynomial in $k$ with guarantee better than $k$ for $k$-hypergraph $b$-matching had been an open problem that we resolve with this work.  Our methods are LP-based, whereas local search seems to give the best known results; however, the bounding techinques used in local search for hypergraph $1$-matching do not seem to readily extend to the hypergraph $b$-matching case.  For example, Arkin and Hassin~\cite{ArkHas98} give a local search $(k-1+\varepsilon)$-algorithm for weighted $k$-hypergraph $1$-matching; however, as a warmup they present a trivial bound of $k$ --- even this trivial bound does not easily extend to the $k$-hypergraph $b$-matching case. 

\paragraph{Other work.}
Pseudo-greedy methods similar to iterated packing have been successfully applied to several packing and coloring problems, including multicommodity flows on trees~\cite{ChekuriEtAl-TALG07}, independent sets in $t$-interval graphs~\cite{BarYehudaEtAl-SICOMP06}, and weighted edge coloring of bipartite graphs~\cite{FeiSin08}.  Iterated packing is a means of obtaining an approximate convex decomposition; Carr and Vempala~\cite{CarrVempala-RSA02} have shown a strong connection between the latter and approximation ratios of LP-based approximation algorithms.


As mentioned earlier, a $2k$-approximation for $k$-hypergraph demand matching is known~\cite{Parekh11}; a better ratio of 3 is possible when $k=2$~\cite{Parekh11}. These nearly match (exactly match, when $k=2$) the best known lower bound of $2k-1$~\cite{BansalEtAl-IPCO10} on the integrality gap of the natural LP relaxation (this construction does not require that $k-1$ is a prime power). Bansal et al.~\cite{BansalEtAl-IPCO10} devised a deterministic $8k$-approximation and a randomized $(ek + o(k))$-approximation for the more general problem of approximating $k$-column-sparse packing integer programs.

Stamoulis very recently introduced the bounded-color matching problem (defined above in the more general hypergraph context) and devised a 2-approximation~\cite{Sta14}.  This result is also based on iterated packing.  Stamoulis observes that the bounded-color matching problem is a special case of $3$-hypergraph $b$-matching.  In fact it is suggested in this paper that a polynomial-time $(k-1+\frac{1}{k})$-approximation for $k$-hypergraph $b$-matching may be possible.  Our work was developed independently, and we observe that our results generalize Stamoulis's results, since the special hypergraph $b$-matching instances obtained by the reduction he suggests are bipartite, and we are able to leverage Theorem~\ref{theorem:2} to give a $k$-approximation for the more general bounded-color $k$-hypergraph $b$-matching problem, which we introduce here.

We will exploit the interplay between LP-relative approximation algorithms and convex decompositions --- an equivalence between the two was shown by Carr \& Vempala~\cite{CarrVempala-RSA02}. The Lavi-Swamy~\cite{LS05} mechanism combines techniques from \cite{CarrVempala-RSA02} with the VCG mechanism.

We give an overview of iterated packing in the next section.  There, we also introduce a structure theorem from \cite{ChaLau11} and its specialization to bipartite instances, versions of which will be used throughout the paper. Next, to further introduce the iterated packing methodology, we give an iterated packing proof of the same result, although it does not run in polynomial time.
This is extended to $b$-matching in Section~\ref{sec:hbm}, which contains our main technical innovations. First an existential proof is given (Algorithm~\ref{alg:bmatchsimple}) and then finally Algorithm~\ref{alg:last} proves Theorems~\ref{theorem:1} and~\ref{theorem:2} constructively. Then in Section~\ref{sec:hdm} we present the proof of Theorem~\ref{theorem:3}, which is based on the local ratio method.


\section{Iterated Packing Overview}\label{sec:hm}

The notion of an approximate convex decomposition is essential to iterated packing, as the latter iteratively builds such a decomposition for a given fractional solution.  Here we present a slightly different notion of an approximate convex decomposition than usually considered.

\begin{definition}
For $\alpha \ge 1,$ define \emph{$\alpha$-convex multipliers} to be any collection of nonnegative reals whose sum is $\alpha$. Likewise, we say that $x$ is an \emph{$\alpha$-convex combination} of the points $\{x^i\}_i$ if there are $\alpha$-convex multipliers $\{\lambda_i\}_i$ so that $x = \sum_i \lambda_i x^i$.
\end{definition}
\noindent
The utility of $\alpha$-convex combinations is that they provide a convenient way to talk about integrality gaps without rescaling as was done in \cite{CarrVempala-RSA02} or \cite{Parekh11}.
\begin{proposition}[\cite{CarrVempala-RSA02}]\label{prop:cvsimple}
If every feasible LP solution for a packing program can be written as an $\alpha$-convex combination of integral feasible solutions, then
its integrality gap is at most $\alpha$.
\end{proposition}
\begin{proof}
We need to show that for any nonnegative weight function $w$, if $x^*$ is the fractional solution that maximizes $w(x^*)$, then there is an integral solution of weight at least $w(x^*)/\alpha$. A random solution from the $\alpha$-convex combination representation of $x^*$, drawing $x^i$ with probability $\lambda_i/\alpha$, has expected weight $\sum_i \frac{\lambda_i}{\alpha} w(x^i) = w(x^*)/\alpha$. So one of the $x^i$ has at least this weight.
\end{proof}
(In fact \cite{CarrVempala-RSA02} also proves an algorithmic converse, used also by the Lavi-Swamy framework~\cite{LS05} underlying Corollary \ref{cor:auctions}.)

We will use Proposition~\ref{prop:cvsimple} as follows: we develop a polynomial-time algorithm to write fractional hypergraph $b$-matchings as $\rho$-convex combinations of feasible integral $b$-matchings. Then by Proposition~\ref{prop:cvsimple}, we get the LP-relative $\rho$-approximation algorithm claimed in Theorems \ref{theorem:1} and \ref{theorem:2}.

In \cite{Parekh11} the idea of iterated packing was introduced. Each iteration, called a \emph{packing step}, updates the current $\alpha$-convex combination to a new one, increasing some terms of the combination on one coordinate.
\begin{definition}[Packing step]
Let us be given an $\alpha$-convex combination $x = \sum_i \lambda_i x^i$ where the $x^i$ are feasible integral solutions, an edge $e$ to pack, and a target value $t \in \R_+$. We may think of a \emph{packing step} as packing the edge $e$ into some of the solutions $x^i$ such that each resulting solution is still feasible and that we have packed $e$ into solutions with a total mass of t, i.e. the sum of corresponding $\lambda_i$ is $t$.  

Let $\chi_e$ be the vector in $\R^E$ with coordinate 1 on $e$ and 0 elsewhere. A \emph{packing step} will replace some $0 \le \lambda'_i \le \lambda_i$ portion of each $x^i$ with $x^i + \chi_e$, where we allow $\lambda'_i > 0$ only when $x^i + \chi_e$ is feasible. Therefore $\sum_i (\lambda_i - \lambda'_i) \cdot x^i + \sum_i \lambda'_i \cdot (x^i + \chi_e)$, the result of the packing step, expresses $x + t\chi_e$ as an $\alpha$-convex combination of integer feasible solutions.  
\end{definition}
For a packing step to actually be feasible, it is clearly both necessary and sufficient that the set $P = \{i \mid x^i + \chi_e \textrm{ feasible}\}$ of solutions into which $e$ can be packed must satisfy $\lambda(P) \ge t$.

For the sake of polynomial-time implementation of our final algorithm, note we can ensure at most one $i$ has $\lambda'_i \notin \{0, \lambda_i\}$ in the above argument, so that each packing step increases the number of terms by at most one. Alternatively we could use Carath\'eodory's theorem which guarantees that any $\alpha$-convex combination can be rewritten as one with at most $d+1$ terms where $d$ is the number of coordinates.

The basic iterated packing formula starts with a fractional solution $x$ in hand and iteratively constructs an integral solution by starting with an empty hypergraph on $V$.  The edges are processed in some order, and for each edge $e$, a packing step is performed on $e$ with a target value of $x_e$.  One key fact about iterated packing is that when a target value is larger, packing is easier, hence iterated packing shows how large fractional values facilitate approximation for packing problems much like iterated packing does for covering problems.  The basic approach may be refined in several directions.  One may start with base integral solution that is non-empty hypergraph.  This was explored in~\cite{Parekh11} to derived an improved approximation for the demand matching problem.  Another improvement is to consider a specific ordering of edges. 

This key idea driving our algorithm is analyzing an ordering of edges which allows us to obtain a polynomial-time algorithm.  Although, as announced in~\cite{Parekh11}, extensions of ideas from ~\cite{Parekh11} may be used to derive an upper bound of $k-1+1/k$ on the integrality gap for the $k$-hypergraph $b$-matching problem, the bound is non-constructive and does not give a polynomial-time algorithm.  We show that by considering an ordering of edges that was first studied by Chan and Lau~\cite{ChaLau11}, we obtain a polynomial-time $k-1+1/k$-approximation.  This ordering is based on vertices of small degree in an extreme point solution, which in turns allows one to argue that there is an edge with large fractional value.  The lemma below shows that we can find a vertex of sufficiently small degree.

Let $\{A_{v,e}\}_{v, e}$ be the 0-1 incidence matrix for our $k$-hypergraph: it has rows for vertices and columns for edges, with at most $k$ ones per column. When $x^*$ is an extreme point solution to the matching LP $\{0 \le x \le 1 \mid Ax \le 1\}$, elementary properties of polyhedra show that the incidence matrix of $\{e \mid 0 < x^*_e < 1\}$ has linearly independent columns. This makes the following lemma useful: it was proven by Chan and Lau for the general case, while the bipartite case follows from generalizing their arguments about the $k$-dimensional case.
\begin{lemma}\label{lemma:count}
If the incidence vectors of $\varnothing \neq E' \subseteq E$ are linearly independent, then some vertex in $(V, E')$ has degree between 1 and $k$. In the bipartite case, the upper bound can be strengthened to $k-1$.
\end{lemma}
\begin{proof}
The first part is a counting argument. The incidence matrix retains its rank if we delete the all-zero rows, leaving only those rows corresponding to the set $V'$ of vertices with nonzero degree. The number of such vertices must satisfy $|V'| \ge |E'|$ or else rank $|E'|$ could not be achieved. Since each column has at most $k$ unit entries, there are at most $k|E'|$ unit entries in the whole matrix. So averaging, some row has at most $k|E'|/|V'| \le k$ nonzeroes, and this gives the desired vertex.

For bipartite hypergraphs, examine the situation in which equality holds. This can only happen if $|E'|=|V'|$ and the matrix has exactly $k$ ones per row and per column. Let $U$ be the subset of vertices so that every hyperedge intersects $U$ exactly once. So, each hyperedge intersects the complement of $U$ exactly $k-1$ times. Therefore, the vector in $\R^V$ with $(-k-1)$ entries in $U$ and unit entries elsewhere is orthogonal to all rows, contradicting that the adjacency matrix has full rank.
\end{proof}
In order to talk about both the general and bipartite cases in a unified way, define
$$\rho:=\begin{cases}k-1+\tfrac{1}{k}&\textrm{ in the general case, and }\\k-1&\textrm{ in the bipartite case.}\end{cases}$$
Additionally, define the degree bound
$$\mu:=\begin{cases}k&\textrm{ in the general case, and }\\k-1&\textrm{ in the bipartite case.}\end{cases}$$

\section{Non-Polynomial Time Algorithm for $k$-Hypergraph Matching}\label{match-ip}

We now give an alternate proof that $k$-hypergraph matching has integrality gap of at most $k+1-\tfrac{1}{k}$. The algorithm behind this proof does not run in polynomial time. However, this section also introduces the notation and steps involved in iterated packing, which we will extend in the next section to get our main result.

\begin{lemma}[\cite{Parekh11}]\label{lemma:validpack}
In $k$-hypergraph matching, a packing step to bring $x$ to $x + t\chi_e$, where $x+t\chi_e$ is a feasible fractional solution, is possible if $\alpha \ge k - (k - 1)t$.
\end{lemma}
\begin{proof}
Let $Q_v$, for each $v \in e$, be the set of solutions $i$ for which $x^i + \chi_e$ is not feasible. We have $\lambda(Q_v) \le 1-t$ since $x+t\chi_e$ is feasible\footnote{In detail, the solutions $x^i$ for $i \in Q_v$ have degree 1 at $v$, so by the definition of a convex combination $(Ax)_v = \lambda(Q_v)$, but $(Ax)_v \le 1 - t$ since, by feasibility, $1 \ge A(x+t\chi_e)_v = (Ax)_v + t$.}.
We need room (disjoint in the worst case) for all such $Q_v$, plus an additional $t$ to pack the new edge in solutions that permit it, giving the bound $k(1-t) + t = k - (k-1)t$.\end{proof}

We can indeed get large coordinates using the following strengthening of Lemma~\ref{lemma:count}.
\begin{lemma}\label{lemma:highvalue}
Any nonzero extreme point solution $x$ to the $k$-hypergraph matching polytope has some fractional coordinate at least $1/\mu$.
\end{lemma}
\begin{proof}
If $x_e = 1$ for any coordinate then we are done, so suppose otherwise. We know from elementary linear algebra that there is a set $V''$ of vertices and a set $E''$ of edges so that $x$ is the unique solution to $x_e = 0, \forall e \notin E''; x(\delta(v))=1, \forall v \in V''$. Then the same counting argument as in Lemma~\ref{lemma:count} (resp.~and the same linear independence in the bipartite case) ensures that some $v \in V''$ is incident on at most $k$ (resp.~$k-1$) edges. Since it has $x(\delta(v))=1$ the $e \in \delta(v)$ maximizing $x_e$ satisfies the lemma.
\end{proof}

Using this, we obtain an iterated packing algorithm for the $k$-hypergraph matching problem, which is displayed as Algorithm~\ref{alg-hmcc}.
Note that this algorithm is presented as a recursive top-down variant of iterated packing, while the basic version in the previous section was presented as a bottom-up algorithm for ease of exposition.  Another more crucial deviation of this algorithm from the basic iterated packing formula is that since our analysis requires an extreme point, we must express each non-extreme solutions as convex combinations of extreme points, and we use that:
\begin{equation}\textrm{A convex combination of $\alpha$-convex combinations is an $\alpha$-convex combination.}\label{eq:combcomb}\end{equation}
In fact this is the reason the algorithm is not guaranteed to run in polynomial time; however, the algorithm does terminate since the number of nonzero coordinates of $x$ decreases in each recursive call.

\begin{algorithm}
\caption{$\textrm{HM}^*(V, \E, x)$ \comment{write $x$ as $\rho$-convex comb.~of 0-1 solutions}}\label{alg-hmcc} \label{fig:alg2}
\begin{algorithmic}[1]
\State If $x = \mathbf{0}$ return the trivial $\rho$-convex combination $\lambda_1 = \rho, x_1 = \mathbf{0}$.
\State If $x$ is not an extreme point solution to $\{x \in \R_+^E \mid Ax \le 1\}$,
\State $\quad$ Write $x$ as a convex combination of extreme point solutions.
\State $\quad$ Recurse on each extreme point and return their result combined via \eqref{eq:combcomb}.
\State Pick $e$ so that $x_e$ is maximized and let $x'$ be $x$ except with $x'_e$ set to zero.
\State Recurse: $(x^i, \lambda_i)_i := \textrm{HM}^*(V, E, x')$.
\State Packing step: pack $x_e$ of $e$ into $(x^i, \lambda_i)_i$ and return the result.
\end{algorithmic}
\end{algorithm}

\begin{proposition}\label{prop:isroom}
Given any LP solution $x$, Algorithm~\ref{fig:alg2} returns an expression of $x$ as a $\rho$-convex combination of integral solutions.
\end{proposition}
\begin{proof}
This follows from Lemmas \ref{lemma:validpack} and \ref{lemma:highvalue}, since if $\mu = k$ we have $k-(k-1)/\mu = k-1+\tfrac{1}{k} = \rho$, and if $\mu = k-1$ we have $k-(k-1)/\mu = k-1 = \rho$.
\end{proof}
\noindent
This completes the non-polynomial time iterated packing proof that the integrality gap for matching is at most $\rho$. Next, we extend it to $b$-matching.

\section{Iterated Packing and $k$-Hypergraph $b$-Matching}\label{sec:hbm}
In this section, which contains the main new iterated packing technique, we build on the ideas from the previous section. We begin with a non-constructive iterated packing algorithm to show that the integrality gap for $k$-hypergraph $b$-matching is at most $\rho$. Then, we move to a constructive version via iterated packing that runs in polynomial time. 

By Observation~\ref{obs:caps}, we assume unit capacities (simple $b$-matching).
We will use the following statement, whose proof is analogous to Lemma~\ref{lemma:highvalue}.
\begin{lemma}
Any nonzero extreme point solution $x$ to the $k$-hypergraph $b$-matching polytope has some fractional coordinate at least $1/\mu$.
\end{lemma}

The naive adaptation of iterated packing (Algorithm~\ref{fig:alg2}) to $b$-matching would involve writing the input as a convex combination of extreme point solutions to $\{x \in [0,1]^E \mid Ax \le b\}$, working with $\alpha$-convex combinations of integer 0-1 solutions to $Ax \le b$.
However, this approach is unworkable. When we try to mimic Lemma~\ref{lemma:validpack}, as $b$ gets larger, we cannot bound $\lambda(Q_v)$ by anything less than 1, giving an approximation ratio of $k$ or worse.

To fix this problem, we will enforce two additional conditions.
One of these conditions, the main driver of the new proof, is that the \emph{strengthened} degree bound $Ax^i \le \lceil Ax \rceil$ must hold in \emph{every} level of the recursion (rather than the unworkable requirement that solutions merely respect the final target degrees). The second condition is that the $\lambda$-mass of solutions meeting this strengthened bound with equality cannot be more than $\langle (Ax)_v \rangle$ (here $\langle \cdot \rangle$ denotes the fractional part), except in the degenerate case that $(Ax)_v$ is integral.
Intuitively (i) balances the number of edges packed at a vertex across the solutions $x^i$, avoiding the trouble that the naive approach would encounter in future iterations, while (ii) helps achieve (i) inductively. A \emph{modified packing step} is a packing step that, given a solution $(x,
\lambda)$ satisfying both of these properties, produces another $(x', \lambda')$ satisfying both of these properties.
Then the definition of the resulting algorithm, Algorithm~\ref{alg:bmatchsimple}, is as follows.

\begin{algorithm}
\caption{$\textrm{H}b\textrm{M}^*(V, \E, x)$ \comment{write $x$ as $\rho$-convex comb.~of special 0-1 solutions}} \label{alg:bmatchsimple}\label{fig:bmatchsimple}
\begin{algorithmic}[1]
\State If $x = \mathbf{0}$ return the trivial $\rho$-convex combination $\lambda_1 = \rho, x_1 = \mathbf{0}$. \comment{as before}
\State If $x$ is not an extreme point solution to $\{y \in [0,1]^E \mid Ay \le \lceil Ax \rceil\}$,
\State $\quad$ Write $x$ as a convex combination of extreme point solutions. \comment{as before}
\State $\quad$ Recurse on each extreme point; return their combination via \eqref{eq:combcomb}. \comment{as before}
\State Pick $e$ so that $x_e$ is maximized; let $x'$ be $x$ with $x'_e$ set to zero. \comment{as before}
\State Recurse: $(x^i, \lambda_i)_i := \textrm{H}b\textrm{M}^*(V, E, x')$. \comment{as before}
\State Modified packing step: pack $x_e$ of $e$ into $(x^i, \lambda_i)_i$ and return the result.
\end{algorithmic}
\end{algorithm}

We will prove by induction that the algorithm succeeds in finding packings meeting both conditions.
\begin{lemma}\label{lemma:modified}For any $0 \le x \le 1$, \textrm{H$b$M$^*$}($V, \E, x$) returns an expression of $x$ as a $\rho$-convex combination of 0-1 $x^i$ that satisfies (i) $Ax^i \le \lceil Ax \rceil$ for each $i$, and (ii) for every $v$ such that $(Ax)_v$ is non-integral, $\lambda(\{i \mid (Ax^i)_v = \lceil Ax_v \rceil\}) \le \langle (Ax)_v \rangle$.
\end{lemma}

For the proof, it is helpful to realize that we use $(Ax)_v$ interchangeably as $x(\delta(v))$, and that it represents the ``degree" of $x$ at $v$. 

\begin{proof}
The base case and the non-extreme case are easy; while the extreme points decomposing a non-extreme solution may have smaller values for $\lceil Ax \rceil$, this does not hurt us. So we only need to deal with the case that $x$ is extreme and nonzero, where $e$ is chosen with $x_e \ge 1/\mu$.

To prove that the modified packing step can always be carried out while satisfying (i) and (ii), we again bound a set of unpackable solutions. Specifically, our goal will be to define sets $Q_v$ for each $v \in e$ such that any packing step that avoids adding $e$ to any of the solutions $\bigcup_{v \in e}Q_v$ will satisfy (i) and (ii) for $x$, and such that the sets $Q_v$ are $\lambda$-small enough that $e$ always has room to be added.

For each $v \in e$, there are three cases, the main distinction being whether $\lceil (Ax')_v \rceil = \lceil (Ax)_v \rceil$. Note that these terms are either equal, or differ by one.
\begin{itemize}
\item
Case (I), $(Ax')_v = 0$. This packing is trivial, set $Q_v = \varnothing$.
\item
 Case (II), $\lceil (Ax)_v \rceil = \lceil (Ax')_v \rceil \neq 0$. Proving (ii) is vacuous when $(Ax)_v$ is integral, and otherwise it follows easily by induction since $\langle (Ax)_v \rangle = \langle (Ax')_v \rangle + x_e$ and at most $x_e$ of $\lambda$-mass of solutions will have its degree increased at $v$. To show (i) is satisfied inductively, just like in Section~\ref{match-ip}, define $Q_v$ to be the set of $i$ with $(Ax^i)_v = \lceil (Ax)_v \rceil$; $e$ can be added to any other $x^i$ without violating the degree constraint. The terms $(Ax)_v$ and $(Ax')_v$ differ by $x_e$ and have the same integer ceiling, so by induction on (ii), we have the bound $\lambda(Q_v) \le \langle x'(\delta(v)) \rangle \le 1-x_e$ showing that $Q_v$ is not too big. This bound will be used later.
\item
Case (III), $\lceil (Ax)_v \rceil  = 1+ \lceil (Ax')_v \rceil$ and $(Ax')_v \neq 0$. Then satisfying (i) at $v$ is easy
 (since all $x^i$ have degree at most $\lceil (Ax')_v \rceil$ at $v$) but we must design $Q_v$ so that (ii) is satisfied after the packing step.

 If $(Ax)_v$ is integral any packing works (we can take $Q_v=\varnothing$), so assume the opposite.
Moreover, when $(Ax')_v$ is integral, by (i) all solutions $x^i$ have degree less than $\lceil (Ax)_v \rceil$ at each $v \in e$, and since we are only packing $x_e = \langle (Ax)_v \rangle$ amount of $e$, (ii) is also satisfied by any possible packing.

 Hence, assume both $(Ax')_v$ and $(Ax)_v$ are non-integral.  If we pack $e$ arbitrarily, the total weight of new solutions with degree $\lceil (Ax)_v \rceil$ at $v$ could be too large to satisfy (ii).
 Therefore, we will define $Q_v$ to exclude some subset of the solutions $Q'_v := \{i \mid (Ax^i)_v = \lceil (Ax')_v \rceil\}$ that could rise to have this degree. We have $\lambda(Q'_v) \le \langle (Ax')_v \rangle$ from (ii) inductively.
 We now define $Q_v$ to be some subset of $Q'_v$ with $\lambda(Q_v) = 1-x_e$. This is not possible if $\lambda(Q'_v) < 1-x_e$ but in this case we just define $Q_v := Q'_v$. Also, even if no subset of $Q'_v$ has $\lambda$-value \emph{exactly} $1-x_e$ we can split\footnote{Splitting means to replace the term $(x^i, \lambda^i)$ with two terms $(x^i, p), (x^i, \lambda^i-p)$ with distributed $\lambda$-mass on the same integer solution $x^i$.} a term of the $\rho$-convex decomposition to achieve this. The point of this $Q_v$ is that, using (ii) inductively, the post-packing total $\lambda$-value of the solutions with degree $\lceil (Ax)_v \rceil$ at $v$ will be at most $\lambda(Q'_v \bs Q_v) \le \langle (Ax')_v \rangle - (1-x_e) =  \langle (Ax)_v \rangle$; the latter equality holds since $(Ax)_v = (Ax')_v + x_e$ and by the hypotheses of this case. So these $Q_v$ allow us to inductively satisfy (i) and (ii), on top of which $\lambda(Q_v) \le 1-x_e$.
\end{itemize}
In all cases, $\lambda(Q_v) \le 1-x_e$. Analogous to Lemma~\ref{lemma:validpack} there is enough room to complete the packing step so long as $\rho \ge x_e + \lambda(\bigcup_{v \in e}Q_v)$. By a union bound this would be implied by $\rho \ge x_e + k(1-x_e)$. This gives the same analysis as before (Proposition~\ref{prop:isroom}) in terms of our bounds on $x_e$ and $\rho$, so the modified packing step succeeds and we are done.
\end{proof}

\subsection{Polynomial-Time Iterated Packing for $k$-Hypergraph $b$-Matching}
Finally, we give our main algorithm. It uses modified packing steps and always maintains a $\rho$-convex combination satisfying the conditions of Lemma~\ref{lemma:modified}. As usual, the core algorithm H$b$M (Algorithm~\ref{fig:last}) operates on solutions where the incidence matrix $A$ is of full column rank.

\begin{algorithm}
\caption{H$b$M($V, \E, x$) \comment{write $x$ as $\rho$-convex comb.~of 0-1 solutions}}\label{alg:last}\label{fig:last}
\begin{algorithmic}[1]
\Require $A$ has its columns linearly independent
\State If $x = \mathbf{0}$ return the trivial $\rho$-convex combination $\lambda_1 = \rho, x_1 = \mathbf{0}$.
\State Pick a vertex $\widehat{v}$ with minimum nonzero degree.
\State Pick $e \in \delta(\widehat{v})$ such that $x_e$ is maximized.
\State Recurse: $(x^i, \lambda_i)_i := \mathrm{H}b\mathrm{M}(V, E \bs \{e\}, x|_{E \bs \{e\}})$.
\State Extend each $x^i$ back to $\R^E$ by setting the $e$-coordinates to 0.
\State Modified packing step: pack $x_e$ of $e$ into $(x^i, \lambda_i)_i$ and return the result.
\end{algorithmic}
\end{algorithm}

\begin{lemma}
If $0 < x < 1$ and the columns of the incidence matrix $A$ are linearly independent, H$b$M expresses $x$ as a $\rho$-convex combination of $0$-$1$ solutions satisfying the same properties as Lemma~\ref{lemma:modified}.
\end{lemma}
\begin{proof}
The proof is very similar to proof of Lemma~\ref{lemma:modified} (except we have linear independence instead of extremeness)
and we therefore re-use its notation and some of the observations therein. Our goal is to show that each modified packing step succeeds.
Write $Q$ for $\bigcup_{v \in e}Q_v$. For the modified packing step to succeed we need $\lambda(Q) + x_e \le \rho$ as before. We will use that $\lambda(Q_v) \le (1-x_e)$ for each $v$, which holds as in Lemma~\ref{lemma:modified}.

The first case we will handle is $|e|<k$. In this case, $\lambda(Q)+x_e \le |e|(1-x_e) +x_e \le (k-1)(1-x_e) + x_e \le k-1 \le \rho$, as needed. So we assume $|e|=k$.

Since Lemma~\ref{lemma:count} applies to our setting, the degree of $\widehat{v}$ is at most $\mu$. The next case we will handle is $x_e \ge 1/\mu$. In this case, $\lambda(Q) + x_e \le k(1-x_e) + x_e
= k - (k-1)x_e \le k - (k-1)/\mu = \rho$
(like the proof of Proposition~\ref{prop:isroom}).
So we may assume $x_e < 1/\mu$.

Likewise, by the definition of $\mu$, we may assume $x(\delta(\widehat{v})) < 1$, since otherwise we fall in to the previous case by our choice of $e$.

Since $x(\delta(\widehat{v})) < 1$, we can get an exact expression for $Q_{\widehat{v}}$ more specific than that given in the proof of Lemma~\ref{lemma:modified}. All solutions $x^i$ in the $\rho$-convex combination have degree 0 or 1 at $v$, and the latter are the ones in $Q_{\widehat{v}}$ (blocking $e$ at $\widehat{v}$), and so $\lambda(Q_{\widehat{v}}) = (Ax')_{\widehat{v}} = (Ax)_{\widehat{v}} - x_e = x(\delta({\widehat{v}})) - x_e$.
This complements the upper bounds $\lambda(Q_v) \le 1-x_e$ that hold for all other $v \in e$ with $v \neq \widehat{v}$. This lets us bound the amount of room needed for the modified packing step:
\begin{align*}x_e + \lambda(Q) &\le x_e + x(\delta({\widehat{v}})) - x_e + (k-1)(1-x_e) \\&\le \mu x_e + (k-1)(1-x_e)
= k-1+(\mu-k+1)x_e \\&\le k-1+(\mu-k+1)/\mu = k-(k-1)/\mu = \rho
\end{align*}
where the middle inequality used $x(\delta({\widehat{v}})) \le \mu x_e$ and the last used $x_e < \frac{1}{\mu}$.
\end{proof}

To complete the proofs of Theorems~\ref{theorem:1} and \ref{theorem:2}, we yet again use the approach of starting with an extreme point solution and fixing its integer part (like Observation~\ref{obs:caps}), recursing only on the residual $b$-matching problem, which has linearly independent rows and $0 < x < 1$.

\section{Application: Bounded-Color $k$-Hypergraph $b$-Matching}\label{sec:BCHM}
We observe that improved approximations for the bounded-color $k$-hypergraph $b$-matching problem, which is defined above Corollary~\ref{cor:BCHM}, follow directly from our results.  The specialization of this problem for the case of matchings in graphs was very recently introduced by Stamoulis~\cite{Sta14}, who gave a 2-approximation (note that Stamoulis had considered only matchings and not $b$-matchings).  This independent result also leverages a variant of iterated packing.  We give a $k$-approximation for the general case of bounded-color $k$-hypergraph $b$-matching, and thus extend the above result to hypergraphs as well as $b$-matchings.

Stamoulis observed that bounded-color matching is a special case of $3$-hypergraph $b$-matching: for each color class $E_i$, add a new vertex $c_i$ with capacity $w_i$.  Now replace each edge $\{u,v\}$ with a hyperedge $\{c_i, u, v\}$.  This precisely models the bounded-color matching problem.  An analogous reduction shows that bounded-color $k$-hypergraph $b$-matching is a special case of standard $(k+1)$-hypergraph $b$-matching.  To obtain our approximation, we simply observe that these special instances are bipartite, as the set $U$ consisting of all the $c_i$ vertices intersects every hyperedge exactly once.  This gives us a $k$-approximation since the instance under consideration is a $(k+1)$-hypergraph.

\section{Application: Allocations}\label{app:auctions}


We will take advantage of the Lavi-Swamy framework~\cite{LS05}, which is a fractional version of the well-known Vickrey-Clarke-Groves (VCG) mechanism. We cannot directly use VCG in this setting, because one of the steps in VCG is to compute the allocation which maximizes the total utility of all players, and this problem is \NP-complete in our setting for $t \ge 2,$ by a reduction from 3-dimensional matching. The main result of Lavi and Swamy is that once we have an \emph{LP-relative $\rho$-approximation algorithm} with respect to the natural LP, we can get a truthful-in-expectation mechanism, which also maximizes the expected overall utility within a factor of $\rho.$ Minimizing this factor means we are coming closer to a VCG-like mechanism, whereas allocating everyone the empty set is truthful but a bad approximation.

First we define the natural LP relaxation for the allocation problem. Let $x^i_S$ be a fractional indicator variable indicating whether player $i$ will win exactly the set $S$ of items. Then the LP requires that each player wins one set of items, and that each item is allocated at most once, fractionally. Write $v^i_S$ as the valuation of player $i$ for set $S$. Altogether the fractional allocation LP is:
\begin{equation}
\max \sum_{i, S} x^i_Sv^i_S : 0 \le x \le 1; \forall i \in [n]: \sum_S x^i_S = 1; \forall s \in [m]: \sum_i \sum_{S : s \in S} x^i_S \le 1.
\tag{\ensuremath{\mathcal{A}}}\label{eq:alloc}
\end{equation}

We assume the input to the mechanism is an explicit list from each bidder, consisting of their valuation for each set upon which they wish to put a positive bid. The number of variables and constraints in the LP is polynomial in the number of such bids. Although for constant $k$, any reasonable bid language or oracle can be used, since the number of sets of size $<k$ is polynomial and we can convert everything to an explicit list.

\begin{definition}An \emph{$\rho$-approximate truthful-in-expectation mechanism} for the allocation problem is a randomized algorithm of the following form. It takes the values $v$ as inputs; its outputs are a valid allocation of items to players together with prices $p_i$ charged to each player $i$. It has the following two properties. First, where $S(i)$ denotes the set of items allocated to player $i$, we have $\sum_i v^i_{S(i)}$ is at least $\sum_i v^i_{T(i)}/\rho$ for every valid allocation $T$. Second, for every fixed $v^{-i}$, a player who gives insincere valuations $\widehat{v}^i$ as their input, resulting in random variables $\widehat{p}, \widehat{S}$ compared to the original ones $p, S$, does not increase their expected net utility:
$$E[v^i_{\widehat{S}(i)} - \widehat{p}_i] \le E[v^i_{S(i)} - p_i].$$
Moreover, $0 \le E[p_i] \le E[v^i_{S(i)}]$ for all $i$.
\end{definition}

\begin{theorem}[Lavi-Swamy~\cite{LS05}]
Given a polynomial-time LP-relative $\rho$-approximation algorithm for an allocation problem, we can obtain a polynomial-time $\rho$-approximate truthful-in-expectation mechanism.
\end{theorem}
However, the allocation problem here is precisely bipartite $k$-hypergraph matching: for each bidder and each set of items they could win, create a set out of them all together, and this set has size at most $1+k-1 = k$; and each such hyperedge contains exactly one bidder, so the hypergraph is indeed bipartite. So our bipartite extension of the Chan-Lau theorem (Section~\ref{sec:hm}) applies and we are done. The LP-relative property is essential; the non-LP relative local search approach from~\cite{FNSW11} cannot be used with~\cite{LS05}.

\section{Local Ratio and $k$-Hypergraph Demand Matching}\label{sec:hdm}
We recommend \cite{BY+04,BYR05,BYRM06} for background on the local ratio method, including its relationship with the primal-dual method.
The heart of the local ratio approach is the following lemma:
\begin{lemma}[Local ratio lemma]
Let $x_{OPT}$ be the (unknown) optimal integral solution. If $w_i \cdot x_{LR} \ge w_i \cdot x_{OPT}$ for all $i$, and $w = \sum_i w_i$, then $w \cdot x_{LR} \ge w \cdot x_{OPT}$, i.e.~$x_{LR}$ is $\alpha$-approximately optimal.
\end{lemma}
Compared with fractional local ratio, we do not start by solving an LP, which is faster. But, we cannot use $x^*$ to guide the algorithm --- we have to ensure an oblivious approximation guarantee that holds against the unknown optimal solution.

In this section we briefly outline a reinterpretation of the $2k$-approximation for $k$-hypergraph demand matching from~\cite{Parekh11} as a local ratio algorithm. Compared with \cite{Parekh11}, the new algorithm will be both simpler and faster (as we solve no LPs). The inspiration for this simplified algorithm is a connection between local ratio algorithm and iterated packing elucidated by Bar-Yehuda et al.~\cite[p.~12]{BarYehudaEtAl-SICOMP06}.

As before, let $A$ be the incidence matrix, and let $A[d]$ be the same matrix but with the column for each $e$ having its entries multiplied by $d_e$. Then an ILP formulation for the hypergraph demand matching problem is to find an integral $x$ maximizing $wx$ subject to $A[d]x \le b$ and $c \ge x \ge 0$.
We will assume that $d_e \le b_v$ whenever $v \in e$. This is without loss of generality for the purposes of approximation, while for bounding the integrality gap this \emph{no-clipping assumption} is needed to even get a constant upper bound (even if $k=1$, a.k.a.~knapsack). 

We use the same basic ideas used in~\cite{Parekh11} but arranged differently.  The crux in our case is to show that for every instance, there is a hyperedge $e$ and a weight function satisfying that any feasible solution is either $2k$-approximately optimal or has room for $e$ to be added. With this (Lemma~\ref{lemma:lr}) and using the local ratio lemma, we can show that Algorithm~\ref{fig:alg1} is a $2k$-approximation algorithm.

\begin{lemma}\label{lemma:lr}Let $e$ be the hyperedge so that $d_{e}$ is minimal. Define a weight function $\widehat{w}$ on all hyperedges by $\widehat{w}_e=1$, and for all other $f$,
\begin{equation}\widehat{w}_f := \sum_{v \in e \cap f} \frac{d_f}{\max\{b_v-d_e, d_e\}}.\label{eq:wf}\end{equation}
Then (i) every feasible solution (whether or not it contains $e$) has value at most $2k$ under $\widehat{w}$, (ii) $\widehat{w}_e \ge 1$, and (iii) any feasible subset of $E \bs \{e\}$ to which $e$ cannot be added has weight at least 1 under $\widehat{w}$.
\end{lemma}


\begin{algorithm}
\caption{HDM($V, \E, d, b, w$) \comment{for hypergraph demand matching}}\label{alg-primal-dual}\label{fig:alg1}
\begin{algorithmic}[1]
\State Pick $e \in \E$ such that $d_e$ is minimum, or return $\varnothing$ if $E = \varnothing$. \label{line-pick-e}
\State Define a new weight function $\widehat{w} \in \R^\E$ via $\widehat{w}_e = 1$ and \eqref{eq:wf} for $f \neq e$.
\State Let $w_e\widehat{w}$ be its scalar multiple by $w_e$, and $w' := w - w_e \widehat{w}$. \comment{note $w'_e=0$}
\State Define $E' := \{e \in E \mid w'_e > 0\}.$ \comment{note $e \not\in E'$}
\State Recurse: $\F' := \mathrm{HDM}(V, E', d, b, w'|_{E'})$.
\State If $\F' \cup \{e\}$ is feasible define $\F := \F' \cup \{e\}$, else define $\F := \F'$.
\State Return $\F$.
\end{algorithmic}
\end{algorithm}

\bibliography{b-matching}
\bibliographystyle{abbrv}

\end{document}